\documentclass[
reprint,
amsmath,amssymb,
superscriptaddress,
aps,
pre,
]{revtex4-1}

\usepackage{graphicx,color}
\usepackage{mathrsfs}   % script fonts
\usepackage{natbib}
\usepackage[breaklinks]{hyperref}
\usepackage{verbatim}
\usepackage{enumerate}
\usepackage{amsthm}

\newcommand{\bidmark}{\rm u}
\newcommand{\inmark}{\rm i}
\newcommand{\outmark}{\rm o}

\newcommand{\kin}{\ensuremath{{k^\text{(i)}}}}
\newcommand{\kout}{\ensuremath{{k^\text{(o)}}}}
\newcommand{\kbid}{\ensuremath{{k^\text{(u)}}}}

\newcommand{\kvec}{\ensuremath{\mathbf{k}}}
\newcommand{\xvec}{\ensuremath{\mathbf{x}}}

\newcommand{\zeroone}{\ensuremath{ \{0,1\} } }
\newcommand{\onthresholdi}{\phi_{{\rm on}, i}}
\newcommand{\offthresholdi}{\phi_{{\rm off}, i}}
\newcommand{\onthreshold}{\phi_{{\rm on}}}
\newcommand{\offthreshold}{\phi_{{\rm off}}}
\newcommand{\avgdeg}{k_{\rm avg}}
\newcommand{\stochf}{\bar{f}}
\newcommand{\smurl}[1]{{\small \url{#1}}}

\newtheorem{mythm}{Lemma}
\newcommand{\dee}[1]{\mbox{d}#1}

\begin{document}

\title{Dynamical influence processes on networks:\\ General theory and applications to social contagion}
%% \title{Dynamical influence processes on networks: \\
%% General theory and applications to social contagion}
%% \title{Dynamics of influence processes on networks:\\
%%   Complete mean-field theory;
%%   the roles of response functions, connectivity, and synchrony;
%%   and applications to social contagion}

\author{Kameron Decker Harris}
\altaffiliation{
  Current address:
  Department of Applied Mathematics,
  University of Washington,
  Seattle, WA 98103 USA
}
\email{kamdh@uw.edu}
\affiliation{
  Department of Mathematics and Statistics,
  Vermont Advanced Computing Core,
  Vermont Complex Systems Center,   
  and Computational Story Lab,
  University of Vermont,
  Burlington, VT 05405 USA
}
\author{Christopher M.\ Danforth}
\email{chris.danforth@uvm.edu}
\affiliation{
  Department of Mathematics and Statistics,
  Vermont Advanced Computing Core,
  Vermont Complex Systems Center,   
  and Computational Story Lab,
  University of Vermont,
  Burlington, VT 05405 USA
}
\author{Peter Sheridan Dodds}
\email{peter.dodds@uvm.edu}\
\affiliation{
  Department of Mathematics and Statistics,
  Vermont Advanced Computing Core,
  Vermont Complex Systems Center,   
  and Computational Story Lab,
  University of Vermont,
  Burlington, VT 05405 USA
}
% alternate:
% \affiliation{
%   Department of Mathematics and Statistics,
%   University of Vermont,
%   Burlington, VT 05405 USA\\ %}
%  %\affiliation{
%   Vermont Advanced Computing Core,
%   University of Vermont,
%   Burlington, VT 05405 USA\\ %}
% %\affiliation{
%   Vermont Complex Systems Center,   
%   University of Vermont,
%   Burlington, VT 05405 USA\\
%   Computational Story Lab,
%   University of Vermont,
%   Burlington, VT 05405 USA
% }

\date{August 28, 2013}

\begin{abstract}
  We study binary state dynamics on a network
where each node acts in response to the average state of its neighborhood.
Allowing varying amounts of stochasticity in both the network and
node responses, we find different outcomes in random and
deterministic versions of the model. 
In the limit of a large, dense network, however,
we show that these dynamics coincide.
We construct a general mean field theory for random networks
and show this predicts that the 
dynamics on the network are a smoothed version of the average
response function dynamics.
Thus, the behavior of the system can range from 
steady state to chaotic depending on the
response functions,
network connectivity,
and update synchronicity.
As a specific example, 
we model the competing tendencies of imitation
and non-conformity by incorporating an off-threshold
into standard threshold models of social contagion.
In this way we attempt to capture important aspects of 
fashions and societal trends.
We compare our theory to extensive simulations
of this ``limited imitation contagion'' model
on Poisson random graphs, 
finding agreement between the mean-field theory
and stochastic simulations.

\end{abstract}

\maketitle

\section{Introduction}

Networks continue to be an exploding
area of research due to their
generality and ubiquity in
physical, biological, technological, and social settings.
Dynamical processes taking place on networks are
now recognized as the most natural description for a number
of phenomena. 
These include 
neuron behavior in the brain \citep{coombes2006a},
cellular genetic regulation \citep{milo2002network},
ecosystem dynamics and stability \citep{bascompte2009disentangling},
and infectious diseases \citep{rohani2010contact}.
This last category, 
the study of biological contagion,
is in many ways similar to {\it social} contagion,
which refers to the spreading of 
ideas, fashions, or behaviors among people
\citep{ugander2012structural,centola2010a}.
This concept underlies the vastly important contemporary area
of viral marketing, 
driven by the ease with which media can be shared and spread 
through social network websites.

In this work,
we present results
for a very general model of 
networked map dynamics,
motivated by models of social contagion.
We will describe our model in a social
context, but it is more general
since it is a type of boolean network \citep{aldana2003a}.
These are closely related to physical models of
percolation \citep{stauffer1994a}
and
magnetism \citep{newman2003a, aldana2003a},
and they have been employed in a number of fields
such as computational neuroscience \citep{schneidman2006a}, 
ecology \citep{campbell2011a}, and others.
Nodes, which in our case represent people, are allowed two possible states.
These could encode rioting or not rioting
\citep{granovetter1978a},
buying a particular style of tie
\citep{granovetter1986a},
liking a band or style of music,
or taking a side in a debate
\citep{galam2008a}.
Each node has a response function,
a map which determines the state the node will take in response
to the average state of nodes in its neighborhood.
The model can thus represent many behaviors,
as long as there are only two mutually exclusive possibilities,
where agents make a choice based on the average of their neighbors'
choices.
\citet{schelling1971a, schelling1973a} and \citet{granovetter1978a}
pioneered the use of threshold response functions in such models
of collective social behavior.
This was based on the intuition that, 
for a person to adopt some new behavior, 
the fraction of the population exhibiting
it might need to exceed some critical value, 
the person's threshold.
These threshold social contagion models,
which are a subclass of our general model,
have already been studied on networks
\citep{watts2002a,dodds2004a}.

In our theoretical analysis,
we focus on the derivation and analysis
of dynamical master equations
that describe the expected evolution
of the system state in the general influence process model.
These master equations are given in 
both exact form and mean-field approximations.
We also show how certain dense network limits lead to the convergence
of the dynamics to the average response function map dynamics.

We then apply the
general theory to a particular
limited imitation contagion model
\citep{dodds2013a}.
Nodes act according to competing tendencies
of imitation and non-conformity.
One can argue that these two ingredients 
are essential to all trends; indeed,
Simmel, in his classic essay ``Fashion'' \citeyearpar{simmel1957a},
believed that these are the main forces 
behind the creation and destruction of fashions.
Our model is not meant to be quantitative, except
perhaps in carefully designed experiments.
Nodes in the model lack memory of their past states,
which is likely an important effect in the adoption of real fashions.
It does capture qualitative features with which we are familiar: 
some trends take off and some do not, and
some trends are stable while others vary wildly through time. 
% This is closely related to the seminal work
% of \citet{schelling1971a} and \citet{granovetter1978a}.

In Section~\ref{sec:model}, 
we define the general model and its
deterministic and stochastic variants. 
In Section~\ref{sec:fixedanalysis},
we provide an analysis of the model when the underlying network is fixed.
In Section~\ref{sec:meanfield}, 
we develop a mean-field theory
of the model on generalized random networks.
In Section~\ref{sec:licmodel}, we consider the model
on Poisson random networks with a specific kind of response function
that reflects the limited imitation
we expect in many social contagion processes.
For this specific case, we compare
the results of simulations and theory. 
Finally, in Section~\ref{sec:conclusions},
we present conclusions
and directions for further research.

\section{General model}
\label{sec:model}

Let $\mathscr{G} = (V, E)$ be a network with
$N=|V|$ nodes,
where $V$ is the node set and $E$ is the edge set.
We let 
$A = A(\mathscr{G})$
denote the {\it adjacency matrix};
entry
$A_{ij}$
is the number of edges from 
node $j$ to node $i$.
Assign each node
$i \in V$ 
a response function
$f_i : [0, 1] \to \{0, 1\}$,
and let
$\xvec(0) \in \zeroone^N$ 
be the vector of initial node states.
At time step $t$,
each node $i$ computes the fraction 
\begin{equation}
  \label{eq:phiexpression}
  \phi_i(t) = \frac{\sum_{j=1}^N A_{ij} x_j (t)}{\sum_{j=1}^N A_{ij}}
\end{equation}
of their neighbors in 
$\mathscr{G}$
who are active and takes the state 
\begin{equation}
  \label{eq:detupdatemap}
  x_i(t+1)
  = f_i \left( \phi_i(t) \right)
\end{equation}
at the next time step. 

The above defines a deterministic dynamical system
given a network and
set of response functions.
We call this a realization of the model
\citep{aldana2003a}.
Each node is in either the 0 or 1 state;
we refer to these as the off/inactive
and on/active states, respectively.
In the context of contagion, these would be the 
susceptible and infected states.
With these binary states, 
our model is a particular kind of 
Boolean network.
These exhibit rich dynamics and have a long history
in the literature.
Unfamiliar readers should consult the review by \citet{aldana2003a}
and references therein.
Note that each node reacts only to the fraction of its neighbors
who are active, rather than the absolute number, 
and the identities of the input nodes do not matter.
Each node's input varies from 0 to 1 in steps 
of $1/k_i$, where 
$k_i = \sum_{j=1}^N A_{ij}$
is node $i$'s degree 
(in-degree if $\mathscr{G}$
is not a simple graph).

The behavior of the model depends strongly on the response functions
$f_i$.
Leaving these undetermined, 
the principle feature of the model is its
neighborhood-averaging property.
Because of local averaging, one might expect
that the dynamics of the network global average activity
might be close to the map dynamics of
the average of the $f_i$.
We show this is the case in dense enough networks in
Sections~\ref{sec:densepoisson}
and  
\ref{sec:mfanalysis}.
This averaging property also introduces an invariance to
the number of inputs a given node receives.

In the rest of this Section, 
we will describe some variations of the basic model
which also differentiate our model
from the Boolean networks extant in the literature.
This is mainly due to the response functions,
but also the type of random network on which the dynamics
take place,
varying amounts of stochasticity
introduced into the
networks and response functions, 
and the possibility of asynchronous updates.

\subsection{The networks considered}

The mean-field analysis in Section~\ref{sec:meanfield} 
is applicable to any network which can be characterized
by its degree distribution.
The vast majority of the theory of
random Boolean networks considers only regular random networks.
Fortunately, such theories are easily generalized to other types
of networks with independently chosen edges,
such as Poisson (Erd\"os-R\'enyi)
and configuration model random networks
\citep{newman2003a,bollobas2001a}.
We develop specific results for
Poisson random networks,
and these are the networks used for the example problem in
Section~\ref{sec:licmodel}.

\subsection{Stochastic variants}

The specific network 
and response functions
determine exactly which behaviors are possible.
These are chosen from some distribution of networks, such as
$\mathcal{G}(N, \avgdeg/N)$
(Poisson random networks on $N$ nodes with edge probability
$\avgdeg/N$), 
and some distribution of response functions.
In the example of Section~\ref{sec:licmodel},
the response functions are parameterized solely by two
thresholds, 
$\onthreshold$ 
and
$\offthreshold$,
so the distribution of response functions is
determined by the joint density 
$P( \onthreshold , \offthreshold )$.
Again, the specific network and response functions
define a realization of the model.
When these are fixed for all time, 
we have, in principle, full knowledge of
the possible model dynamics. 
Given an initial condition $\xvec(0)$,
the dynamics $\xvec(t)$ 
are deterministic and known for all $t \geq 0$.
As for all finite Boolean networks,
the system dynamics are eventually periodic, 
since the state space $\{0,1\}^N$ is finite \citep{aldana2003a}.

We allow for randomness in two parts of the model:
the network and response functions.
The network and responses are each either
fixed for all time or resampled each time step.
Taking all possible combinations yields
four different designs 
(see Table~\ref{tab:modeldesigns}).
If the dynamics are stochastic in any way,
the system is no longer
eventually periodic.
Fluctuations at the node level enable
a greater exploration of state space, and the behavior is 
comparable to that of the general class of discrete-time maps.
Roughly speaking, the mean-field theory we develop in 
Section~\ref{sec:meanfield} becomes more accurate as we introduce
more stochasticity.

In this paper, the network and response functions are either
fixed for all time or resampled every time step. 
One could tune smoothly between the two extremes by
introducing rates at which these reconfigurations occur.
These rates are inversely related to quantities that
behave like temperature,
one for the network and another for the response functions.
Holding a quantity fixed
corresponds to zero temperature, 
since there are no fluctuations.
Any randomness is quenched.
The stochastic and rewired cases correspond to 
high or infinite temperature,
because reconfigurations occur every time step.
This is an annealed version of the model.

\begin{table}
  \centering
  \begin{tabular}{l|c|c}
    & {\bf R}ewiring network & {\bf F}ixed network \\
    \hline
    {\bf P}robabilistic response & P-R & P-F \\
    {\bf D}eterministic response & D-R & D-F
  \end{tabular}
  \caption[The four different ways we implement the model,
    corresponding to differing amounts of quenched randomness.]
  {The four different ways we implement the model,
    corresponding to differing amounts of quenched randomness. 
    These are the combinations of fixed or rewired networks and
    probabilistic or deterministic response functions.
    In the thermodynamic limit of the rewired versions, 
    where the network and response functions
    change every time step, the mean-field theory 
    (Sec.~\ref{sec:meanfield})
    is exact.}
  \label{tab:modeldesigns}
\end{table}

\subsubsection{Rewired networks}
\label{sec:stochnet}

First, the network itself can change every time step.
This is the rewiring (R), as opposed to fixed (F),
network case.
For example, 
we could draw a new network from $\mathcal{G}(N, \avgdeg/N)$
every time step. 
This amounts to rewiring the links while keeping
the degree distribution fixed, 
and it is alternately known as a mean 
field, annealed, or random mixing variant
as opposed to a 
fixed network or quenched model
\citep{aldana2003a}.
% Statistical physicists also use the words
% ``annealed'' (for stochastic) and
% ``quenched'' (for fixed)
% \citep[see][for an explanation and a short history of the terms]{aldana2003a}.  
\subsubsection{Probabilistic responses}
\label{sec:stochfunc}

Second, the response functions can change every time step.
This is the probabilistic (P),
as opposed to the 
deterministic (D), response function case.
For our social contagion example,
there needs to be a well-defined distribution 
$P( \onthreshold , \offthreshold )$
for the thresholds.
For large $N$,
this amounts to having a single
response function, the expected response function
\begin{equation}
  \label{eq:fstoch}
  \stochf (\phi) = \int \dee \onthreshold 
  \int \dee \offthreshold \; 
  P( \onthreshold , \offthreshold ) 
  f(\phi; \onthreshold, \offthreshold) .
\end{equation}
We call $\stochf: [0,1] \to [0,1]$ the {\it probabilistic response function}.
Its interpretation is the following.
For an updating node with a fraction $\phi$ of
active neighbors at the current time step,
then, at the next time step, the node assumes the 
active state with probability 
$\stochf(\phi)$ 
and the inactive state with probability 
$1 - \stochf(\phi)$.

\subsection{Update synchronicity}

Finally, we introduce a parameter $\alpha$ for the probability that
a given node updates. 
When $\alpha = 1$, all nodes update each
time step, and the update rule is said to be synchronous.
When $\alpha \approx 1/N$, only one node is expected to update with each
time step, and the update rule is said to be effectively asynchronous.
This is equivalent to a randomly ordered sequential update.
For intermediate values, $\alpha$ is the expected fraction of
nodes which update each time step.

\section{Fixed networks}

\label{sec:fixedanalysis}

Consider the case where the response functions and
network are fixed (D-F), but the update may be synchronous or asynchronous.
Extend the definition of 
$x_i(t)$ 
to now be the 
{\it probability} that node $i$ is in the active state at time $t$.
Note that this agrees with our previous definition
as the {\it state} of node $i$
when
$x_i(t) = 0$ or 1.
% and let 
% $f_i (\phi) = f_i(\phi; \onthresholdi, \offthresholdi)$.
Then the $x_i$ follow the master equation
\begin{equation}
  \label{eq:quenchedmap}
  x_i (t+1) = 
  \alpha
  f_i \left( 
    \frac{ \sum_{j=1}^N A_{ij} x_j (t)}{\sum_{j=1}^N A_{ij}}
  \right)
  + (1-\alpha) x_i(t),
\end{equation}
which can be written in matrix-vector notation as 
\begin{equation}
  \label{eq:quenchedmapmatrix}
  \xvec (t+1) = \alpha \, \mathbf{f} \left( T \xvec(t) \right) 
  + (1-\alpha) \xvec(t).
\end{equation}
Here $T = D^{-1}A$ is sometimes called the
transition probability matrix 
(since it also occurs in the context of a random walker),
$D = \text{diag}( k_i)$ is the diagonal degree matrix, and 
$\mathbf{f} = (f_i)$
\footnote{When there are isolated nodes, 
  their degrees are 0, and
  thus the denominator in the master equation 
  is zero and $D$ is singular.
  If the initial network contains isolated nodes, 
  we set all entries in the corresponding rows of $T$ to zero.
}.
If $\alpha = 1$, then
$\xvec(t) \in \{0,1\}^N$ and
we recover the fully deterministic
response function dynamics given by 
\eqref{eq:phiexpression} and \eqref{eq:detupdatemap}.

% For the rest of this section, we will focus on
% the case of quenched Erd\H{o}s-R\'enyi random networks.
% These are also called the $G(n,p)$ model.
% Here, the adjacency matrix is an $N \times N$ matrix
% with Bernoulli random variables for the off-diagonal entries.
% In the case of undirected networks, $A$ is constrained to be symmetric.
% Specifically, $A_{ij} \in \zeroone$ for all $i, j$,
% $A_{ij} = 1$ with probability $p$ when $i \neq j$, and
% $A_{ij} = 0$ with probability 1 for $i = j$. 
% We can let $z = p N$ be the expected 
% degree for a node and parametrize the network by $z$ instead of $p$.

\subsection{Asynchronous limit}

Here, we show that when $\alpha \approx 1/N$, 
time is effectively continuous and the dynamics can be described by an 
ordinary differential equation. This is similar to the analysis of 
\citet{gleeson2008a}.
Consider Eqn.~\eqref{eq:quenchedmapmatrix}. Subtracting $\xvec(t)$ from
both sides and setting $\Delta \xvec(t) = \xvec(t+1)-\xvec(t) $ and 
$\Delta t = 1$
yields
\begin{equation}
  \frac{\Delta \xvec(t)}{\Delta t} = 
  \alpha \left( \mathbf{f}( T \xvec(t) ) - \xvec(t) \right).
\end{equation}
Since $\alpha$ is assumed small, the right hand side is small, and 
thus $\Delta \xvec(t)$ is also small.
Making the continuum approximation 
$\dee \xvec(t)/ \dee t \approx \Delta \xvec(t)/\Delta t$
yields the differential equation
\begin{equation}
  \label{eq:diffeq}
  \frac{\dee \xvec}{\dee t} = 
  \alpha \left( \mathbf{f}( T \xvec ) - \xvec \right).
\end{equation}
The parameter $\alpha$ sets the time scale for the system.
Below we see that, from their form, 
similar asynchronous, continuous time limits
apply to the dynamical equations in the densely
connected case, 
Eqn.~\eqref{eq:densemap},
and in the mean-field
theory, 
Eqns.~\eqref{eq:edgemap} and~\eqref{eq:nodemap}.

\subsection{Dense network limit for Poisson random networks}

\label{sec:densepoisson}

Mathematical random graph and random matrix theories
often deal with condensation results,
where quantities of interest, such as adjacency matrices,
become overwhelmingly concentrated around some typical
value in a limit.
These limits are dense in the sense that 
$\avgdeg \to \infty$ 
in the thermodynamic limit
$N \to \infty$.
The mean field theory of
Section~\ref{sec:meanfield}
applies to sparse graphs, 
which have finite 
$\avgdeg$
in the thermodynamic limit.
The following result is particular to Poisson random networks,
however, similar results should hold for other random networks
with corresponding dense limits.

Define the normalized Laplacian matrix as
$\mathcal{L} \equiv I - D^{-1/2} A D^{-1/2}$,
where $I$ is the identity \citep{west2001a}.
So $T = D^{-1/2} (I- \mathcal{L}) D^{1/2}$. 
% We see that $T$ is similar to $I - \mathcal{L}$, and the associated
% change of basis is a rescaling of axis $i$ by $\sqrt{k_i}$.
Let
$\mathbf{1}_N$
denote the length-$N$ column vector of ones.
\citet{oliveira2009a}
has shown that when
$\avgdeg$
grows at least as fast as
$\log N$,
there exists a typical
normalized Laplacian matrix
$\mathcal{L}^\text{typ} 
=  I - \mathbf{1}_N \mathbf{1}_N^T /N$
such that the actual 
$\mathcal{L} \approx \mathcal{L}^\text{typ}$.
In this limit, the degrees should also be approximately uniform,
$k_i \approx \avgdeg$ for all $i \in V$,
since the coefficient of variation of degrees
vanishes in the Poisson model as
$\avgdeg \to \infty$.

If we use the approximations
$\mathcal{L} \approx \mathcal{L}^\text{typ}$
and
$D \approx \text{diag} \, (\avgdeg)$,
then
$T \approx T^\text{typ} = \mathbf{1}_N \mathbf{1}_N^T / N$.
Since
\begin{equation*}
  T^\text{typ} \xvec(t) = 
  \sum_{i=1}^N x_i (t)/N \equiv \phi(t),
\end{equation*}
$T^\text{typ}$ operating on the state gives the network average activity,
denoted
$\phi$
without any subscript.
Using the previous approximations in 
Eqn.~\eqref{eq:quenchedmapmatrix}
and averaging over all nodes, 
we find
\begin{equation}
  \phi(t+1) = \alpha  \stochf (\phi(t)) + (1-\alpha) \phi(t) 
  \equiv \Phi(\phi(t); \alpha)
  \label{eq:densemap}
\end{equation}
in the large $N$ limit.
This requires
the average of nodes' individual response functions 
$\sum_{i=1}^N f_i/N$
to converge in a suitable sense to the probabilistic
response function $\stochf$, Eqn.~\eqref{eq:fstoch}.
% This amounts to assuming a law of large numbers for the response functions.
Note that $\alpha$ tunes between the probabilistic response
function
$\Phi(\phi;1) = \stochf(\phi)$
and the line
$\Phi(\phi; 0) = \phi$.
Also, the fixed points of $\Phi$ are fixed points of
$\stochf$, but their stability will depend on $\alpha$.

We conclude that 
when the network is dense, it ceases to affect the dynamics, since
each node sees a large number of other nodes. 
The network is effectively the complete graph.
In this way we recover
the map models of \citet{granovetter1986a},
which are derived for a well-mixed population.

\section{Mean-field theory}

\label{sec:meanfield}

Making a mean-field calculation refers to replacing
the complicated interactions among many particles
by a single interaction with some effective external field.
There are analogous techniques for understanding network
dynamics. 
Instead of considering the $|E|$ interactions
among the $N$ nodes, network mean-field theories derive
self-consistent expressions for the overall behavior of the network
after averaging over large sets of nodes.
These have been fruitful in the study of random Boolean
networks \citep{derrida1986a} and
can work well when networks are non-random \citep{melnik2011a}.

We derive a mean-field theory, in the thermodynamic limit,
for the dynamics of the general model
by blocking nodes according to their degree class. 
This is equivalent to nodes retaining their degree but rewiring 
edges every time step.
The model is then part of the well-known class of random mixing
models with non-uniform contact rates.
Probabilistic (P-R) and deterministic (D-R) response functions
result in equivalent behavior for these random mixing models.
The important state variables end up being the active density of 
{\em stubs},
i.e.\ half-edges or node-edge pairs.
In an undirected network without degree-degree correlations, 
the state is described by a single variable $\rho(t)$. 
In the presence of correlations we must
introduce more variables 
$\left\{ \rho_k(t) \right\}$ 
to deal with the relevant degree classes.

\subsection{Undirected networks}

\label{sec:mfundir}

To derive the mean-field equations in the simplest 
case---undirected, uncorrelated random networks---consider
a degree $k$
node at time $t$. 
The mean-field hypothesis states that the probability
a given stub is active is uniform across all nodes
and equal to $\rho$.
Then the probability that an average
node is in the active state at time $t+1$
given a density $\rho$ of active stubs is 
\begin{equation}
  \label{eq:Fkmap}
  F_k ( \rho; \stochf ) = \sum_{j=0}^{k} \binom{k}{j} 
  \rho^j \left(1-\rho \right)^{k-j} 
  \stochf \left(j/k \right),
\end{equation}
where each term in the sum counts the contributions 
from having 
$j = 0, 1, \ldots, k$ 
active neighbors.
The probability of 
choosing a random stub which ends at a degree $k$ node is 
$q_k = k p_k/\avgdeg$
in an uncorrelated random network \citep{newman2003a}.
This is sometimes called 
the edge-degree distribution.
So if all of the nodes update synchronously, the active density
of stubs at $t+1$ will be
\begin{equation}
  \label{eq:gmap}
  g(\rho ; p_k, \stochf)
  = \sum_{k=1}^{\infty} q_k F_k \left( \rho; \stochf \right)
  = \sum_{k=1}^{\infty} \frac{k p_k}{\avgdeg} F_k \left( \rho; \stochf \right).
\end{equation}
Finally, if each node only updates with probability $\alpha$,
we have the following map for the density of active stubs:
\begin{equation}
  \begin{aligned}
    \rho(t+1) &= \alpha \, g \left( \rho(t); p_k, \stochf \right)  
    + (1-\alpha) \rho(t)\\
    &\equiv G \left( \rho(t); p_k, \stochf, \alpha \right) .
  \end{aligned}
  \label{eq:edgemap}
\end{equation}
By a similar argument, the active density of nodes is given by
\begin{equation}
  \label{eq:nodemap}
  \begin{aligned}
    \phi(t+1) &= \alpha \, h \left(\rho(t); p_k, \stochf \right) +
    \left( 1-\alpha \right) \phi(t)\\
    &\equiv H \left( \rho(t), \phi(t); p_k, \stochf, \alpha \right),
  \end{aligned}
\end{equation}
where
\begin{equation}
  h \left( \rho; p_k, f \right) = 
  \sum_{k=0}^{\infty} p_k F_k \left( \rho; \stochf \right) .
\end{equation}
Note that the stub or edge-oriented state variable $\rho$
contains all of the dynamically important information,
rather than the node-oriented variable $\phi$.
This is because nodes can only influence each other 
through edges,
so the number of active edges
is a more important 
measure of the network activity than
the number of active nodes.
We derived these equations in terms of stubs, 
meaning that
$\rho$
actually keeps track of both the node and edge information.
Eqns.~\eqref{eq:gmap} and \eqref{eq:edgemap}
can be interpreted as a branching process for
the density of active stubs.

\subsection{Analysis of the map equation and another dense limit}
\label{sec:mfanalysis}

Here we study the properties of the mean field
maps $G$ and $H$,
Eqns.~\eqref{eq:edgemap} and \eqref{eq:nodemap},
which in turn depend on $F_k$, Eqn.~\eqref{eq:Fkmap}.
The function 
$F_k(\rho; \stochf)$ 
is known in polynomial approximation theory
as the $k$th Bernstein polynomial (in the variable $\rho$) of 
$\stochf$ \citep{phillips2003a}.
Bernstein polynomials have important applications
in computer graphics due to their 
``shape-preserving properties'' \citep{pena1999a}.
The Bernstein operator $\mathbb{B}_k$ takes 
$\stochf \mapsto F_k$. 
This is a linear, positive operator which preserves convexity
for all $k$ and exactly interpolates the endpoints 
$\stochf(0)$
and 
$\stochf(1)$. 
Immediate consequences include that
each $F_k$ is a smooth function and the $m$th derivatives
$F_k^{(m)}(x) \to \stochf^{(m)}(x)$ 
where
$\stochf^{(m)}(x)$ 
exists.
For $\stochf$ concave down,
such as the tent or logistic maps, 
then $F_k$ is concave down for all
$k$ and
$F_k$ increases to $\stochf$ 
($F_k \nearrow \stochf$)
as $k \to \infty$.
This convergence is typically slow. 
Importantly, $F_k \nearrow \stochf$ implies that 
$g\left( \rho; p_k, \stochf \right) \leq \stochf$  
% and
% $G(\rho; p_k, f, \alpha) \leq \Phi(\rho; \alpha, f)$
for any degree distribution $p_k$.
%% \begin{enumerate}
%% \item linear $f(x) = ax + b$ has $F_k = f$ and $g = f$.
% \item 
% \item there is no Gibbs phenomenon when considering $F_k$ as 
%   an interpolant of $f$
% \end{enumerate}

In some cases,
the dynamics of the undirected mean-field theory
given by 
$\rho(t+1)=G \left(\rho(t) \right)$, 
Eqn.~\eqref{eq:edgemap},
are effectively those of the map $\Phi$, 
from the dense limit Eqn.~\eqref{eq:densemap}.
% This is   the shape-preserving properties
% of Bernstein polynomials.
We see that  $g$, Eqn.~\eqref{eq:gmap},
can be seen as the expectation of a sequence
of random functions $F_k$
under the edge-degree distribution $q_k$.
Indeed, this is how it was derived.
From the convergence of the $F_k$'s,
we expect that
$g(\rho; p_k, \stochf) \approx \stochf(\rho)$ 
if the average degree $\avgdeg$ 
is ``large enough'' and 
the edge-degree distribution has a ``sharp enough'' peak about $\avgdeg$
(we will clarify this soon). 
Then as $\avgdeg \to \infty$,
the mean-field coincides with the
dense network limit we found for Poisson random networks, 
Eqn.~\eqref{eq:densemap}.
A sufficient condition for
this kind of convergence
is the same that we used in justifying 
the uniform degree approximation in
Section~\ref{sec:densepoisson}:
The coefficient of variation of the degree distribution
must vanish as $\avgdeg \to \infty$.
Equivalently, the standard deviation $\sigma(\avgdeg)$
of the degree distribution must be $o(\avgdeg)$.
In Appendix~\ref{ax:limitthm} we prove this as Lemma~\ref{thm:1}.

In general, if the original degree distribution $p_k$ is
characterized by having mean $\avgdeg$, variance $\sigma^2$, and 
skewness $\gamma_1$, then the edge-degree distribution $q_k$
will have mean $\avgdeg + \sigma^2/\avgdeg$ and variance
$\sigma^2 [ 1 + \gamma_1 \sigma/\avgdeg - (\sigma/\avgdeg)^2]$.
Considering the behavior as $\avgdeg \to \infty$, we can conclude that
requiring 
$\sigma = o(\avgdeg)$ 
and
$\gamma_1 = o(1)$ 
are sufficient
conditions on $p_k$ to apply Lemma~\ref{thm:1}.
Poisson degree distributions ($\sigma = \sqrt{\avgdeg}$ and
$\gamma_1 = \avgdeg^{-1/2}$) fit
these criteria.
Heavy-tailed families of distributions, in general, do not.

The fact that we can take a dense limit
in the mean-field model and find the same result using
rigorous random matrix theory is worth noting.
In general, mean-field models are not rigorously justified.
For finite $N$, 
the quenched dynamics,
which we know are deterministic and eventually periodic, 
are very different from the annealed dynamics,
which we will show can be chaotic.
The equivalence of the two limits indicates that
quenched and annealed dynamics become indistinguishable
as $N \to \infty$.
We believe that the approach to such a singular limit should reveal
interesting discrepancies between the two models.

\subsection{Generalized random networks}
\label{sec:mixednets}

In more general random networks, 
nodes can have both undirected and directed incident edges.
We denote node degree by a vector
$\kvec = (\kbid, \kin, \kout)^T$  
(for undirected, in-, and out-degree)
and write the degree distribution as
$p_\kvec \equiv P(\kvec)$.
There may also
be correlations between node degrees.
We encode correlations of this type by 
the conditional probabilities
\begin{gather*}
  p^{(\bidmark)}_{\kvec,\kvec'} \equiv
  P( \kvec, \mathrm{undirected} | \kvec') \\
  p^{(\inmark)}_{\kvec,\kvec'} \equiv
  P( \kvec, \mathrm{incoming} | \kvec') \\
  p^{(\outmark)}_{\kvec,\kvec'} \equiv
  P( \kvec, \mathrm{outgoing} | \kvec'),
\end{gather*}
the probability that an edge starting at a degree $\kvec'$
node ends at a degree $\kvec$ node and is, respectively,
undirected, incoming, or outgoing relative to the destination degree
$\kvec$ node. 
We introduced this convention in a series of papers 
\citep{payne2011a, dodds2011a}.
These conditional probabilities can also be defined
in terms of the joint distributions of node types connected by
undirected and directed edges.
We omit a detailed derivation, since it
is similar to that in Section~\ref{sec:mfundir}
and similar to the equations for the time evolution of
a contagion process \citep[][Eqns.~(13--15)]{payne2011a}
\citep[see also][]{gleeson2007a}.

\begin{widetext}
The result is a coupled system of equations for the
density of active stubs which now may depend on
node type ($\kvec$) and
edge type (undirected or directed):
\begin{equation}
\begin{aligned}
  \rho_{\kvec}^{(\bidmark)}(t+1) &= 
   (1 - \alpha) \rho_{\kvec}^{(\bidmark)}(t)
   + \alpha \sum_{\kvec'} p^{(\bidmark)}_{\kvec,\kvec'}
   \sum_{j_u=0}^{\kbid'} \sum_{j_i=0}^{\kin'}
   \binom{\kbid'}{j_u} \binom{\kin'}{j_i}  \\
   & \qquad \times 
  \left[ \rho_{\kvec'}^{(\bidmark)}(t) \right]^{j_u}
  \left[1 - \rho_{\kvec'}^{(\bidmark)}(t) \right]^{(\kbid'-j_u)} 
  \left[ \rho_{\kvec'}^{(\inmark)}(t) \right]^{j_i}
  \left[1 - \rho_{\kvec'}^{(\inmark)}(t) \right]^{(\kin'-j_i)}
  \\
  & \qquad \times
  \stochf \left( \frac{j_u + j_i}{\kbid' + \kin'} \right) ,
\end{aligned}
\label{eq:mfgenbid}
\end{equation}
\begin{equation}
\begin{aligned}
  \rho_{\kvec}^{(\inmark)}(t+1) &=
  (1 - \alpha) \rho_{\kvec}^{(\inmark)}(t) 
  + \alpha \sum_{\kvec'} p^{(\inmark)}_{\kvec,\kvec'}  
  \sum_{j_u=0}^{\kbid'} \sum_{j_i=0}^{\kin'}
  \binom{\kbid'}{j_u} \binom{\kin'}{j_i} 
  \\
  & \qquad \times
  \left[ \rho_{\kvec'}^{(\bidmark)}(t) \right]^{j_u}
  \left[1 - \rho_{\kvec'}^{(\bidmark)}(t) \right]^{(\kbid'-j_u)} 
  \left[ \rho_{\kvec'}^{(\inmark)}(t) \right]^{j_i}
  \left[1 - \rho_{\kvec'}^{(\inmark)}(t) \right]^{(\kin'-j_i)} \\
  & \qquad \times
  \stochf \left( \frac{j_u + j_i}{\kbid' + \kin'} \right).
\end{aligned}
\label{eq:mfgendir}
\end{equation}
The active fraction of nodes at a given time is
\begin{equation}
\begin{aligned}
  \phi(t+1) &= (1 - \alpha) \phi(t)
  + \alpha \sum_{\kvec} p_{\kvec} 
  \sum_{j_u=0}^{\kbid} \sum_{j_i=0}^{\kin}
  \binom{\kbid}{j_u} \binom{\kin}{j_i}  \\
  & \qquad \times
  \left[ \rho_{\kvec}^{(\bidmark)}(t) \right]^{j_u}
  \left[1 - \rho_{\kvec}^{(\bidmark)}(t) \right]^{(\kbid-j_u)} 
  \left[ \rho_{\kvec}^{(\inmark)}(t) \right]^{j_i}
  \left[1 - \rho_{\kvec}^{(\inmark)}(t) \right]^{(\kin-j_i)} \\
  & \qquad \times \stochf \left( \frac{j_u + j_i}{\kbid + \kin} \right) .
\end{aligned}
\label{eq:mfgennodes}
\end{equation}
\end{widetext}
Because these expressions are very similar to the undirected case,
we expect similar convergence properties to those in
Sec.~\ref{sec:mfanalysis}.
However, an explicit investigation of this convergence
is beyond the scope of the current paper.

\section{Limited imitation contagion model}

\label{sec:licmodel}

As a motivating example of these networked map dynamics,
we study an extension of the classical threshold models 
of social contagion
\citep[such as][among others]{schelling1971a, schelling1973a, 
  granovetter1978a, watts2002a,dodds2004a}.
In threshold models,
a node becomes active if the active fraction of its friends
surpasses its threshold.
What differentiates our limited imitation contagion
model from the standard models is that
the response function includes an off-threshold, above which
the node takes the inactive state.
We assign 
each node
$i \in V$ 
an on-threshold
$\onthresholdi$
and an off-threshold 
$\offthresholdi$,
requiring
$0 \leq \onthresholdi \leq \offthresholdi \leq 1$.
Node $i$'s response function 
$f_i (\phi_i) = f_i(\phi_i; \onthresholdi, \offthresholdi)$
is 1 if 
$\onthresholdi \leq \phi_i \leq \offthresholdi$ 
and 0 otherwise.
See Figure~\ref{fig:onoffthresh} for an example on-off
threshold response function.

This is exactly the model of 
\citet{granovetter1986a}, but on a network.
We motivate this choice with the following
\citep[also see][]{granovetter1986a}.
(1) Imitation: the active state becomes favored as the fraction of active 
neighbors surpasses the on-threshold (bandwagon effect).
(2) Non-conformity: the active state is eventually less favorable 
with the fraction of active neighbors past the off-threshold
(reverse bandwagon, snob effect).
(3) Simplicity: in the absence of any raw data of ``actual''
response functions, which are surely highly context-dependent and 
variable, we choose arguably the simplest
deterministic functions which capture
imitation and non-conformity.

A crucial difference between our model and many related
threshold models is that, in those models, 
an activated node can never reenter the susceptible state.
\citet{gleeson2007a} call this the permanently active property
and elaborate on its importance to their analysis.
Annealed or quenched models with the permanently active property
have monotone dynamics.
The introduction of the off-threshold builds in a mechanism for node
deactivation. 
Because nodes can now recurrently transition
between on and off states,
the deterministic dynamics can exhibit
a chaotic transient (as in random Boolean networks \citep{aldana2003a}),
and the long time behavior
can be periodic with potentially high period.
With stochasticity, the dynamics can be truly chaotic.

\begin{figure}
  \centering
  \includegraphics[width=7.5cm, height=7cm]{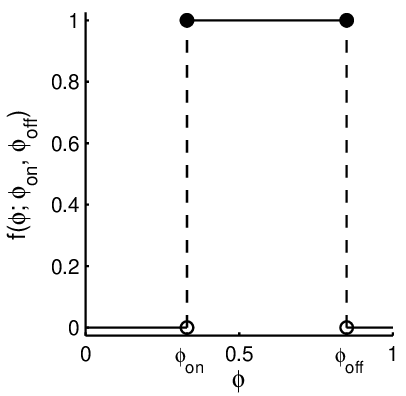}
  \caption[An example on-off threshold response function.]{
    \label{fig:onoffthresh}
    An example on-off threshold response function. 
    Here, 
    $\onthreshold = 0.33$ 
    and
    $\offthreshold = 0.85$. 
    The node will ``activate'' if 
    $\onthreshold \leq \phi \leq \offthreshold$, 
    where $\phi$
    is the fraction of its neighbors who are active. 
    Otherwise it takes the ``inactive'' state.
  }
\end{figure}

The networks we consider are 
Poisson random networks from $\mathcal{G}(N, \avgdeg/N)$.
The thresholds
$\onthreshold$ and $\offthreshold$ 
are distributed uniformly
on $[0, 1/2)$ and $[1/2, 1)$, respectively.
This distribution results in the probabilistic response function
(see Figure~\ref{fig:tentmap})
\begin{equation}
  \label{eq:tentmap}
  \stochf (\phi) = \left\{
    \begin{array}{lrl}
      2\phi \hfill& \mathrm{if} & 0 \leq \phi < 1/2, \\
      2-2\phi \hfill & \mathrm{if} & 1/2 \leq \phi \leq 1 .
    \end{array}
  \right.
\end{equation}
The tent map is a well-known
chaotic map of the unit interval \citep{alligood1996a}.
We thus expected that
the limited imitation model with this
probabilistic response function
to exhibit similarly interesting behavior.

\begin{figure}
  \centering
  \includegraphics[width=7.5cm,height=7cm]{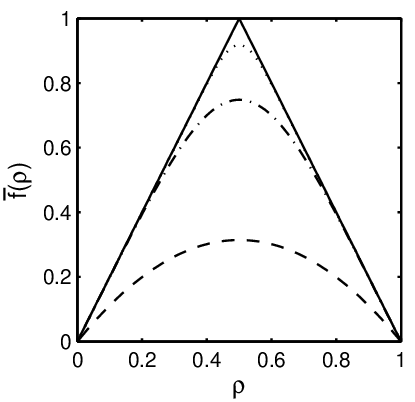}
  \caption[The tent map probabilistic response function $\stochf(\rho)$,
  Eqn.~\eqref{eq:tentmap}.]{
    \label{fig:tentmap}
    The tent map probabilistic response function
    $\stochf(\rho)$,
    Eqn.~\eqref{eq:tentmap}, 
    used in the limited imitation contagion model.
    This is compared to the edge maps
    $g(\rho;\avgdeg) = g(\rho; p_k, \stochf)$, 
    Eqn.~\eqref{eq:gmap}, 
    with 
    $\avgdeg = 1, 10, 100$
    (dashed, dot-dashed, and dotted lines).
    These $p_k$ are Poisson distributions with mean $\avgdeg$.
    As $\avgdeg$ increases, $g(\rho;\avgdeg)$ increases to 
    $\stochf(\rho)$.}
\end{figure}

\subsection{Analysis of the dense limit}
\label{sec:dense}

When the network is in the dense limit (Section~\ref{sec:densepoisson}), 
the dynamics follow
$\phi(t+1) = \Phi( \phi(t); \alpha )$, where
\begin{equation}
  \begin{aligned}
  \Phi(\phi; \alpha) &= \alpha \stochf(\phi) + (1-\alpha) \phi \\
  &=
  \left\{
    \begin{array}{lrl}
      (1+\alpha) \phi & \mathrm{if} & 0 \leq \phi < 1/2, \\
      (1- 3\alpha)\phi + 2\alpha &  \mathrm{if} & 1/2 \leq \phi \leq 1 .
    \end{array}
  \right.
  \end{aligned}
  \label{eq:densetent}
\end{equation}
Solving for the fixed points of $\Phi(\phi; \alpha)$, 
we find one at $\phi=0$ and
another at $\phi = 2/3$. When $\alpha < 2/3$, the 
nonzero fixed point is attracting for all initial conditions
except $\phi = 0$.
When $\alpha = 2/3$, $[1/2, 5/6]$ is an 
interval of period-2 centers. 
Any orbit will eventually land on one of these period-2 orbits.
When $\alpha > 2/3$, this interval
of period-2 centers ceases to exist,
and more complicated behavior ensues. 
Figure~\ref{fig:bifurcdense} shows the bifurcation diagram 
for $\Phi(\phi; \alpha)$.
From the bifurcation diagram,
the orbit appears to cover dense subsets of the unit interval
when $\alpha > 2/3$. 
The bifurcation diagram appears like
that of the tent map 
(not shown; see \cite{alligood1996a, dodds2013a})
except the branches to the right of the first bifurcation point
are separated here by the interval of period-2 centers.

\begin{figure}
  \centering
  \includegraphics[width=8.5cm, height=7cm]{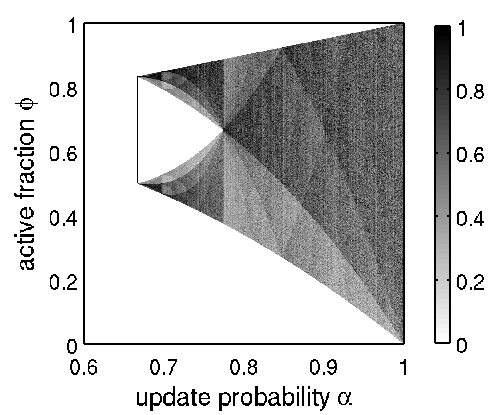}
  \caption[Bifurcation diagram for the dense map $\Phi(\phi;\alpha)$, 
  Eqn.~\eqref{eq:densetent}.]
  {Bifurcation diagram for the dense map $\Phi(\phi; \alpha)$,
    Eqn.~\eqref{eq:densetent}.
    This was generated by iterating the map at 1000
    $\alpha$ values between 0 and 1. 
    The iteration was carried out with
    3 random initial conditions for
    10000 time steps each, discarding the first 1000.
    The $\phi$-axis contains 1000 bins and the invariant density,
    shown by the grayscale value, is normalized
    by the maximum for each $\alpha$. With $\alpha < 2/3$,
    all trajectories go to the fixed point at $\phi = 2/3$.
  }
  \label{fig:bifurcdense}
\end{figure}

\subsubsection*{The effect of conformists, an aside}

Suppose some fraction $c$ of the population is made up of 
individuals without any off-threshold 
(alternatively, each of their off-thresholds $\offthreshold=1$). 
These individuals are conformist or purely pro-social
in the sense that they are content with
being part of the majority. 
For simplicity, assume $\alpha = 1$.
The map $\Phi(\phi; c) = 2 \phi$ for $0 \leq \phi < 1/2$
and $2 - 2(1-c)\phi$ for $1/2 \leq \phi \leq 1$.
If $c > 1/2$, then the equilibrium at 2/3 is stable. Pure conformists, then,
have a stabilizing effect on the process. 
We expect a similar effect when the network is not dense.

\subsection{Mean-field}

Here we mention the methods which were used
to compute the mean-field maps derived in 
Section~\ref{sec:meanfield}.
% Note that the active edge fraction 
% $\rho(t) \approx \phi(t)$,
% the active node fraction.
% This is because Poisson random networks are highly regular, with
% $q_k = k p_k / \avgdeg = p_{k-1} \approx p_k$.
% Thus the mean-field dynamics for active edge density 
% are effectively the same as for active node density.
% \todo{remove previous? seems out of place}
In this specific example,
we can write the degree-dependent map 
$F_k(\rho ; \stochf)$ 
in terms of incomplete regularized beta functions
$I_z(a,b)$ \citep{dlmf2012}. 
Since $\stochf$ is understood to be the tent map, 
we will write 
$F_k(\rho; \stochf) = F_k(\rho)$.
We find that
\begin{equation}
  \label{eq:Fkmapbetafn_bod}
   F_k (\rho) = 
   2 \rho -
   4 \rho I_{\rho} (M, k-M),
\end{equation}
% \begin{subequations}
%   \label{eq:Fkmapbetafn_bod}
%   \begin{align}
%     F_k (\rho) = & \, 2 \rho -
%     4 \rho I_{\rho} (M, k-M) \nonumber \\
%     & \qquad \qquad \qquad + 2 I_{\rho} (M+1, k-M)
%     \label{eq:Fkmapbetafn1_bod}\\
%     = & (2 - 2 \rho) -  2 I_{1-\rho} (k-M, M+1)\nonumber \\
%     & \qquad \qquad \qquad  + 4 \rho I_{1-\rho}(k-M, M) ,
%     \label{eq:Fkmapbetafn2_bod}
%   \end{align}
% \end{subequations}
where we have let $M=\lfloor k/2 \rfloor$ for clarity 
($\lfloor \cdot \rfloor$ and
$\lceil \cdot \rceil$ are the floor and
ceiling functions).
The details of this derivation are given 
in Appendix~\ref{ax:betafn}.

The map $g(\rho; p_k, \stochf)$ 
is parameterized here by the network parameter $\avgdeg$,
since $p_k$ is fixed as a Poisson distribution with mean $\avgdeg$ 
and $\stochf$ is the tent map, and
we write it as simply $g(\rho; \avgdeg)$. 
To evaluate $g(\rho; \avgdeg)$,
we compute $F_k(\rho)$ using Eqn.~\eqref{eq:Fkmapbetafn_bod}
and constrain the sum in 
Eqn.~\eqref{eq:gmap} to values of $k$ with
$\lfloor \avgdeg-3 \sqrt{\avgdeg} \rfloor 
\leq k \leq 
\lceil \avgdeg+3 \sqrt{\avgdeg} \rceil$.
This computes contributions to 
within three standard deviations of the average degree in
the network, requiring only $O (\sqrt{\avgdeg})$ evaluations of 
Eqn.~\eqref{eq:Fkmapbetafn_bod}. 
The representation in  
Eqn.~\eqref{eq:Fkmapbetafn_bod} allows for 
quick numerical evaluation of $F_k( \rho )$ for any $k$,
which we performed in MATLAB.

In Figure~\ref{fig:tentmap},  
we show
$g(\rho; \avgdeg)$ 
for 
$\avgdeg = 1, 10,$ and 100.
We confirm the conclusions of Section~\ref{sec:mfanalysis}:
$g(\rho; \avgdeg)$ 
is bounded above by 
$\stochf(\rho)$, 
and 
$g(\rho; \avgdeg) \nearrow \stochf(\rho)$
as
$\avgdeg \to \infty$. 
Convergence is slowest at $\rho = 1/2$, 
where the kink exhibited by the tent map
has been smoothed out by the
effect of the Bernstein operator.

\subsection{Simulations}

We performed stochastic
simulations of the limited imitation model
for the D-F, P-F, and P-R designs, 
in the abbreviations of Table~\ref{tab:modeldesigns}.
See Supplemental Material at 
[URL will be inserted by publisher] 
for the Python code.
Unless otherwise noted, $N = 10^4$.
For all of the bifurcation diagrams,
the first 3000 time steps were considered transient and discarded,
and the invariant density of $\rho$ 
was calculated from the following 1000 points.
For plotting purposes, the invariant density was normalized by its
maximum at those parameters. 
For example, in Figure~\ref{fig:bifurcdense} we plot
$P(\phi | \alpha) / \max_\phi P( \phi | \alpha)$ 
rather than the raw density 
$P(\phi | \alpha)$.

To compare the mean-field theory to those simulations,
we numerically iterated the edge map
$\rho(t+1) = G \left( \rho(t); \avgdeg, \alpha \right)$
for different values of $\alpha$ and $\avgdeg$.
We then created bifurcation
diagrams of the possible behavior in the mean-field
as was done for the simulations.

\subsection{Results}

\label{sec:results}

\begin{figure}
  \centering
  \includegraphics[width=8.5cm, height=6cm]
  {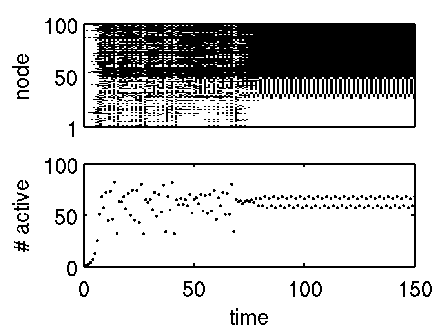}
  \caption[Deterministic (D-F) dynamics on a small network.]
  {
    Deterministic (D-F) dynamics on
    a small network. 
    Here, $N = 100$ and $\avgdeg = 17$. 
    The upper plot shows
    individual node states as a raster plot
    (black = active), sorted by their
    eventual level of activity.
    The lower plot shows the total number 
    of active nodes over time. 
    We see that the contagion takes off, 
    followed by a transient period of unstable behavior
    until time step 80, 
    when the system enters a macroperiod-4 orbit. 
    Note that individual nodes exhibit 
    different microperiods 
    (explained in Sec.~\ref{sec:results}).
    % On the right, 
    % we show the network itself with the initial seed node
    % in black in the lower right.
  }
  \label{fig:smalldet}
\end{figure}

To provide a feel for the deterministic dynamics,
we show the result of running the D-F model
on a small network in Figure~\ref{fig:smalldet}.
Here, $N =100$ and $\avgdeg=17$.
Starting from a single initially active node at
$t = 0$, the active population
grows monotonically over the next 6 time steps. 
From $t=6$ to $t=80$,
the transient time, the active fraction varies
in a similar manner to the dynamics in the stochastic and mean-field cases.
After the transient, the state collapses into a period-4 orbit. 
We call the overall period of the system its ``macroperiod,''
while individual nodes may exhibit different ``microperiods.''
Note that the macroperiod is the lowest common multiple of the individual
nodes' microperiods.
In Figure~\ref{fig:smalldet}, we observe microperiods 1, 2, and 4
in the timeseries of individual node activity.
In other networks, we have observed up to macroperiod 240
\citep{dodds2013a}.
A majority of the nodes end up frozen 
in the on or off state,
with approximately $20\%$ of the nodes exhibiting cyclical behavior
after collapse.
The focus of this paper has been the analysis of the on-off threshold
model, and the D-F
case has not been as amenable to analysis as the 
stochastic cases.
We offer a deeper examination through simulation
of the deterministic case in \cite{dodds2013a}.

\begin{figure}
  \centering
  \includegraphics[width=8.5cm]{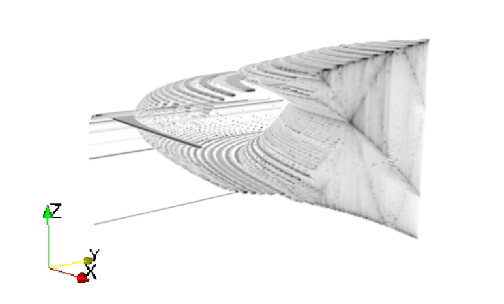}\\
  \includegraphics[width=8.5cm]{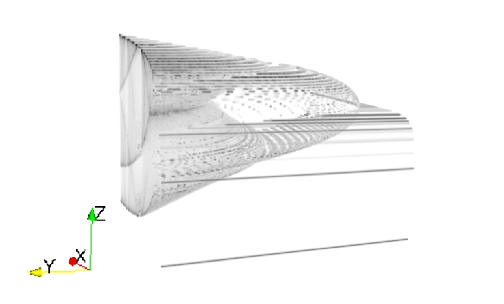}\\
  \includegraphics[width=8.5cm]{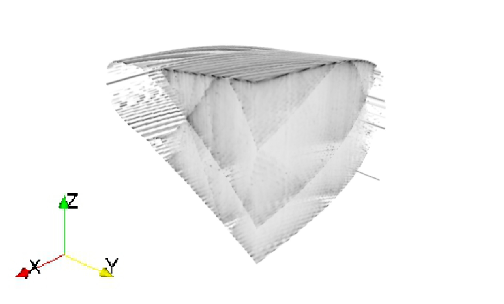}\\
  \caption[The 3-dimensional bifurcation diagram computed 
  from the mean-field theory.]
  {The 3-dimensional bifurcation diagram computed 
    from the mean-field theory.
    The axes 
    X = average degree $\avgdeg$, 
    Y = update probability $\alpha$, and 
    Z = active edge fraction $\rho$. 
    The discontinuities of the surface
    are due to the limited resolution of our simulations.
    See Figure~\ref{fig:bifurcslice_eqn} for the parameters used.
    This was visualized using Paraview.
    See Supplemental Material at 
    [URL will be inserted by publisher] 
    for the underlying data.}
  \label{fig:3dbifurc}
\end{figure}

\begin{figure*}
  \centering
  \includegraphics[width=17.5cm, height=9.5cm]
  {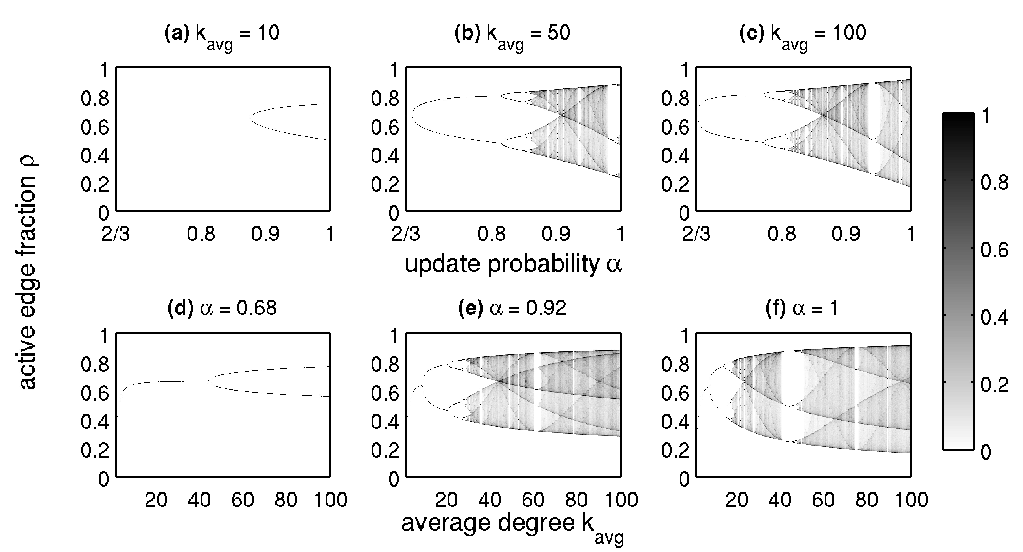}
  \caption[Mean-field theory bifurcation diagram slices 
  for various fixed values of $\avgdeg$ and $\alpha$.]{
    Mean-field theory bifurcation diagram slices 
    for various fixed values of $\avgdeg$ and $\alpha$.
    The top row (a--c) shows slices for fixed $\avgdeg$.
    As $\avgdeg \to \infty$, the $\avgdeg$-slice bifurcation diagram 
    asymptotically approaches the bifurcation diagram for the 
    dense map, Figure~\ref{fig:bifurcdense}. Note that the location of
    the first period-doubling bifurcation point approaches 2/3,
    and the bifurcation diagram more closely resembles
    Fig.~\ref{fig:bifurcdense},
    as $\avgdeg \to \infty$.
    The bottom row (d--f) shows slices for fixed $\alpha$.
    The resolution of the simulations was
    $\alpha = 0.664, 0.665, \ldots, 1$, $\avgdeg=1, 1.33, \ldots, 100$,
    and $\rho$ bins were made for 1000 points between 0 and 1.
  }
  \label{fig:bifurcslice_eqn}
\end{figure*}
\begin{figure*}
  \centering
  \includegraphics[width=17.5cm, height=9.5cm]
  {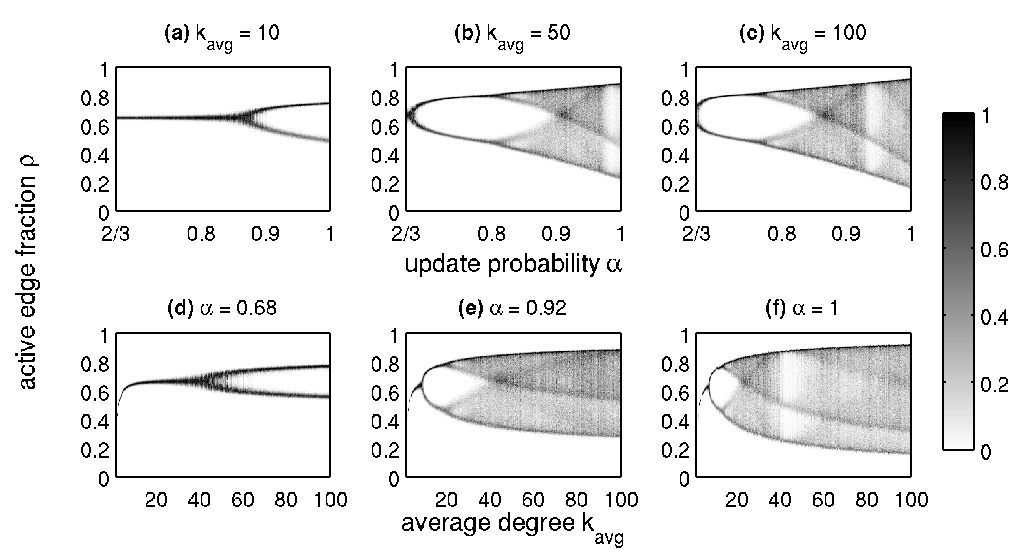}
  \caption[Bifurcation diagram from fully stochastic (P-R) simulations.]
  {
    Bifurcation diagram from fully stochastic (P-R) simulations,
    made in the same way as Figure~\ref{fig:bifurcslice_eqn}.
    The bifurcation structure of these stochastic simulations
    matches that of the mean-field theory 
    (Figure~\ref{fig:bifurcslice_eqn}), albeit with some blurring.
  }
  \label{fig:bifurcslice_sim}
\end{figure*}

We explore the mean-field dynamics by examining the limiting
behavior of the active edge fraction $\rho$
under the map 
$G \left(\rho; \avgdeg, \alpha \right)$.
We simulated the map dynamics
for a mesh of points in the $(\avgdeg, \alpha)$ plane.
We plot the 3-dimensional 
(3-d)
bifurcation structure of the mean-field theory in 
Figure~\ref{fig:3dbifurc}.
We also show 2-d bifurcation plots for
fixed $\avgdeg$ and $\alpha$ slices through this volume in
Figures~\ref{fig:bifurcslice_eqn} and~\ref{fig:bifurcslice_sim}.
For more visualizations of this bifurcation structure,
including movies of the bifurcation diagram as the parameters
are dialed and individual node dynamics,
see 
Supplemental Material at [URL will be inserted by publisher].
In all cases, the invariant density of $\rho$
is normalized by its maximum for that $(\avgdeg, \alpha)$ pair
and indicated by the grayscale value.

The mean-field map dynamics
exhibit period-doubling bifurcations
in both parameters $\avgdeg$ and $\alpha$. Visualizing the bifurcation
structure in 3-d (Figure~\ref{fig:3dbifurc}) shows 
interlacing period-doubling cascades in the two parameter dimensions.
These bifurcations are more clearly resolved when we take slices
of the volume for fixed parameter values.
The mean-field theory (Figure~\ref{fig:bifurcslice_eqn})
closely matches the P-R simulations (Figure~\ref{fig:bifurcslice_sim}).
The first derivative 
$\frac{\partial G}{\partial \rho} (\rho; \avgdeg, \alpha)
< \frac{\partial \Phi}{\partial \rho} (\rho; \alpha)$ 
for any finite $\avgdeg$, 
so the bifurcation point 
$\alpha = 2/3$ 
which we found for the dense map $\Phi$
is an upper bound for the first bifurcation point of $G$.
The actual location of the first bifurcation point depends on $\avgdeg$, 
but $\alpha =2/3$ becomes more accurate for higher $\avgdeg$ 
(it is an excellent approximation in Figures~\ref{fig:bifurcslice_eqn}c
and~\ref{fig:bifurcslice_sim}c, where $\avgdeg=100$). 
When $\alpha =1$, the first bifurcation point occurs at $\avgdeg \approx 7$.

The bifurcation diagram slices resemble each other and evidently
fall into the same universality class as the logistic map
\citep{feigenbaum1978a,feigenbaum1979a}.
This class contains all 1-d maps with a single,
locally-quadratic maximum. Due to the properties of the
Bernstein polynomials, $F_k(\rho; \stochf)$ will universally
have such a quadratic maximum for any concave down, continuous $\stochf$ 
\citep{phillips2003a}. 
So this will also be 
true for 
$g(\rho; \avgdeg, \stochf)$
with 
$\avgdeg$ finite, and we see that $\avgdeg$
partially determines the amplitude of that maximum in 
Figure~\ref{fig:tentmap}. 
Thus $\avgdeg$ acts as a bifurcation parameter. 
The parameter $\alpha$ tunes between 
$G \left(\rho; \avgdeg, 1 \right) = g(\rho; \avgdeg, \stochf)$ 
and 
$G \left(\rho; \avgdeg, 0 \right) = \rho$, 
so it has a similar effect.
Note that the tent map $\stochf$ and the dense limit map
$\Phi$ are kinked at their maxima, so their bifurcation
diagrams are qualitatively different from those of the mean-field.
The network, by locally averaging the node interactions,
causes the mean-field behavior to fall
into a different universality class than
the individual response function map.

\section{Conclusions}
\label{sec:conclusions}

We have described a very general class of
synchronous or asynchronous, binary state
dynamics occurring on networks.
We obtained an exact master equation and
showed that, when random networks are sufficiently dense,
the networked dynamics approach those
of the fully-connected case.
We developed a mean-field theory and
found that it also predicted the same limiting behavior.
The convergence of the mean-field map to the 
average response function is related to 
the Bernstein polynomials,
allowing us to employ many previous results
in order to analyze the mean-field map equation.
We also extended those mean-field equations
to correlated random networks.
We expect that a rigorous mathematical justification can be given
for the mean-field theory of the general influence model
on annealed graphs using a condensation theorem for
configuration model networks.

The general model we describe was motivated by 
the limited imitation model of social contagion.
We see that including an 
aversion to total conformity results in more
complicated, even chaotic dynamics,
as opposed to the simple spreading behavior typically seen
in the single threshold case.
The theory developed for the general case successfully
captured the behavior of the stochastic network dynamics.
We have focused on the rich structure of bifurcations
as the two parameters, 
update synchronicity $\alpha$
and average degree $\avgdeg$,
were varied.
We see that the universality class of the dynamics matches
those of the logistic map.
Using the mean-field theory, 
we can understand this as a result of the smoothing 
effect of the Bernstein polynomials on the tent map average 
response function.
However, this universality class will appear for any unimodular,
concave down stochastic response function.

The deterministic case, which we have barely touched on,
merits further study \citep[see][]{dodds2013a}.
In particular, we would like to 
characterize the distribution of periodic sinks, how the 
collapse time scales with system size, and how similar
the transient dynamics are to the mean-field dynamics.

Furthermore, the model should be tested on realistic networks.
These could include power law or small world random networks,
or real social networks gleaned from data.
One possibility would be to compare data such as
food choices \citep{pachucki2011a}
or Facebook likes \citep{ugander2012structural}
to the model behavior.
In a manner similar to \citet{melnik2011a}, 
one could evaluate the accuracy of the mean-field theory 
for real networks.

Finally, the ultimate usefulness of these social
models relies on a
better understanding of social dynamics themselves. 
Characterization of people's
``real'' response functions is therefore critical
\citep[some work has gone in this direction; see][]
{centola2010a, centola2011a, romero2011a, ugander2012structural}. 
Comparison of model output to large data sets, such
as observational data from social media or online experiments,
is an area for further experimentation.
This might lead to more complicated context- and history-dependent models.
As we collect more data and refine experiments, 
the eventual goal of quantifiably predicting social behavior, 
including fashions and trends,
seems achievable.

\section*{Acknowledgments}
We thank Thomas Prellberg, Leon Glass,
Joshua Payne and Joshua Bongard for discussions and suggestions,
along with the anonymous referees.
We are grateful for computational resources provided by the 
Vermont Advanced Computing Core supported by NASA (NNX 08A096G). 
KDH was supported by VT-NASA EPSCoR and a Boeing Fellowship; 
CMD was supported by NSF grant DMS-0940271; PSD was supported
by NSF CAREER Award \#0846668.
This work is based on the Master's Thesis of KDH.

\appendix

\section{Proof of Lemma 1}

\label{ax:limitthm}

\begin{mythm}
  \label{thm:1}
  For $k \geq 1$, let $f_k$ be continuous real-valued 
  functions on a compact domain $X$
  with $f_k \to f$ uniformly.
  Let $p_k$ be a probability mass function on $\mathbb{Z}^+$
  parameterized by its mean $\mu$ and with standard deviation
  $\sigma(\mu)$, assumed to be $o(\mu)$.
  Then,
  \begin{equation*}
    \label{eq:Top_lim}
    \lim_{\mu \to \infty} \left( \sum_{k=0}^\infty p_k f_k \right) = f .
  \end{equation*}
\end{mythm}

\begin{proof}
  Suppose $0 \leq a < 1$ and let $K = \lfloor \mu-\mu^a \rfloor$. Then,
  \begin{equation}
    g = \sum_{k=0}^{\infty} p_k f_k 
    = \sum_{k=0}^{K} p_k f_k + \sum_{k=K+1}^\infty p_k f_k.
    \label{eq:splitsum}
  \end{equation}
  Since $f_k \to f$ uniformly as $k \to \infty$,
  for any $\epsilon > 0$ we can choose $\mu$
  large enough that
  \begin{equation}
    \label{eq:epsilon}
    |f_k(x) - f(x)| < \epsilon
  \end{equation} 
  for all $k > K$ and all $x \in X$. 
  Without loss of generality, assume that $|f_k| \leq 1$ for all $k$. Then,
  \begin{equation*}
    |g-f| \leq \left( \frac{\sigma}{\mu^a} \right)^2 + \epsilon .
  \end{equation*}
  The $\sigma / \mu^a$ term is a consequence of the 
  Chebyshev inequality \citep{bollobas2001a}
  applied to the first sum in \eqref{eq:splitsum}.
  Since $\sigma$ grows sublinearly in $\mu$, this
  term vanishes for some $0 \leq a <1$
  when we take the limit $\mu \to \infty$.
  The $\epsilon$ term comes from using \eqref{eq:epsilon}
  in the second sum in \eqref{eq:splitsum}, 
  and it can be made arbitrarily small.
\end{proof}

\section{Beta function representation of $F_k$}
\label{ax:betafn}

We now show how,
when $\stochf$ is the tent map \eqref{eq:tentmap},
the map 
$F_k(\rho; \stochf)$
can be written in terms of incomplete regularized beta functions.
First, use the piecewise form of 
Eqn.~\eqref{eq:tentmap} to write
\begin{align}
  F_k(\rho) &= 
  \sum_{j=0}^M \binom{k}{j} \rho^j (1-\rho)^{k-j} \left( \frac{2j}{k} \right) 
  \nonumber \\
  &\quad 
  + \sum_{j=M+1}^{k} \binom{k}{j} \rho^j (1-\rho)^{k-j} 
  \left( 2-\frac{2j}{k} \right) \nonumber \\
  &= 2 - 2 \rho 
  - 2 \sum_{j=0}^M \binom{k}{j} \rho^j (1-\rho)^{k-j} \nonumber \\
  &\quad 
  + \left( \frac{4}{k} \right)
  \sum_{j=0}^M \binom{k}{j} \rho^j (1-\rho)^{k-j} j .
  \label{eq:Fk1}
\end{align}
We have used the fact that the binomial distribution
$\binom{k}{j} \rho^j(1-\rho)^{k-j}$
sums to one and has mean $k \rho$.
For $n \leq M$, we have the identity
\begin{equation}
  \label{eq:factorialmoments}
  \sum_{j=0}^{M} (j)_n \binom{k}{j} \rho^j (1-\rho)^{k-j} = %&= \nonumber \\
  % \sum_{j=n}^{M} (k)_n \binom{k-n}{j-n} \rho^j (1-\rho)^{k-j} &= \nonumber \\
  % \rho^n (k)_n \sum_{j=0}^{M-n} \binom{k-n}{j} \rho^j (1-\rho)^{k-n-j} &= \nonumber\\
  % & \quad 
  \rho^n (k)_n I_{1-\rho}( k-M, M-n+1)
\end{equation}
where $I_{x}(a,b)$ is the regularized incomplete beta function and
$(k)_n = k (k-1) \cdots (k - (n-1))$ is the falling factorial
\citep{winkler1972a, dlmf2012}. 
This is an expression for the partial (up to $M$) 
$n$th factorial moment of the binomial distribution with parameters
$k$ and $\rho$. Note that when $n=0$ we recover the well-known expression
for the binomial cumulative distribution function. 
From Eqns.~\eqref{eq:Fk1} and \eqref{eq:factorialmoments},
we arrive at Eqn.~\eqref{eq:Fkmapbetafn_bod}.

\bibliographystyle{apsrev}
\bibliography{library}

\begin{thebibliography}{40}
\expandafter\ifx\csname natexlab\endcsname\relax\def\natexlab#1{#1}\fi
\expandafter\ifx\csname bibnamefont\endcsname\relax
  \def\bibnamefont#1{#1}\fi
\expandafter\ifx\csname bibfnamefont\endcsname\relax
  \def\bibfnamefont#1{#1}\fi
\expandafter\ifx\csname citenamefont\endcsname\relax
  \def\citenamefont#1{#1}\fi
\expandafter\ifx\csname url\endcsname\relax
  \def\url#1{\texttt{#1}}\fi
\expandafter\ifx\csname urlprefix\endcsname\relax\def\urlprefix{URL }\fi
\providecommand{\bibinfo}[2]{#2}
\providecommand{\eprint}[2][]{\url{#2}}

\bibitem[{\citenamefont{Coombes et~al.}(2006)\citenamefont{Coombes, Doiron,
  Josi{\'c}, and Shea-Brown}}]{coombes2006a}
\bibinfo{author}{\bibfnamefont{S.}~\bibnamefont{Coombes}},
  \bibinfo{author}{\bibfnamefont{B.}~\bibnamefont{Doiron}},
  \bibinfo{author}{\bibfnamefont{K.}~\bibnamefont{Josi{\'c}}},
  \bibnamefont{and}
  \bibinfo{author}{\bibfnamefont{E.}~\bibnamefont{Shea-Brown}},
  \bibinfo{journal}{Philosophical Transactions of the Royal Society A:
  Mathematical, Physical and Engineering Sciences}
  \textbf{\bibinfo{volume}{364}}, \bibinfo{pages}{3301} (\bibinfo{year}{2006}).

\bibitem[{\citenamefont{Milo et~al.}(2002)\citenamefont{Milo, Shen-Orr,
  Itzkovitz, Kashtan, Chklovskii, and Alon}}]{milo2002network}
\bibinfo{author}{\bibfnamefont{R.}~\bibnamefont{Milo}},
  \bibinfo{author}{\bibfnamefont{S.}~\bibnamefont{Shen-Orr}},
  \bibinfo{author}{\bibfnamefont{S.}~\bibnamefont{Itzkovitz}},
  \bibinfo{author}{\bibfnamefont{N.}~\bibnamefont{Kashtan}},
  \bibinfo{author}{\bibfnamefont{D.}~\bibnamefont{Chklovskii}},
  \bibnamefont{and} \bibinfo{author}{\bibfnamefont{U.}~\bibnamefont{Alon}},
  \bibinfo{journal}{Science} \textbf{\bibinfo{volume}{298}},
  \bibinfo{pages}{824} (\bibinfo{year}{2002}).

\bibitem[{\citenamefont{Bascompte}(2009)}]{bascompte2009disentangling}
\bibinfo{author}{\bibfnamefont{J.}~\bibnamefont{Bascompte}},
  \bibinfo{journal}{Science} \textbf{\bibinfo{volume}{325}},
  \bibinfo{pages}{416} (\bibinfo{year}{2009}).

\bibitem[{\citenamefont{Rohani et~al.}(2010)\citenamefont{Rohani, Zhong, and
  King}}]{rohani2010contact}
\bibinfo{author}{\bibfnamefont{P.}~\bibnamefont{Rohani}},
  \bibinfo{author}{\bibfnamefont{X.}~\bibnamefont{Zhong}}, \bibnamefont{and}
  \bibinfo{author}{\bibfnamefont{A.}~\bibnamefont{King}},
  \bibinfo{journal}{Science} \textbf{\bibinfo{volume}{330}},
  \bibinfo{pages}{982} (\bibinfo{year}{2010}).

\bibitem[{\citenamefont{Ugander et~al.}(2012)\citenamefont{Ugander, Backstrom,
  Marlow, and Kleinberg}}]{ugander2012structural}
\bibinfo{author}{\bibfnamefont{J.}~\bibnamefont{Ugander}},
  \bibinfo{author}{\bibfnamefont{L.}~\bibnamefont{Backstrom}},
  \bibinfo{author}{\bibfnamefont{C.}~\bibnamefont{Marlow}}, \bibnamefont{and}
  \bibinfo{author}{\bibfnamefont{J.}~\bibnamefont{Kleinberg}},
  \bibinfo{journal}{PNAS} \textbf{\bibinfo{volume}{109}}, \bibinfo{pages}{5962}
  (\bibinfo{year}{2012}).

\bibitem[{\citenamefont{Centola}(2010)}]{centola2010a}
\bibinfo{author}{\bibfnamefont{D.}~\bibnamefont{Centola}},
  \bibinfo{journal}{Science} \textbf{\bibinfo{volume}{329}},
  \bibinfo{pages}{1194} (\bibinfo{year}{2010}).

\bibitem[{\citenamefont{Aldana et~al.}(2003)\citenamefont{Aldana, Coppersmith,
  and Kadanoff}}]{aldana2003a}
\bibinfo{author}{\bibfnamefont{M.}~\bibnamefont{Aldana}},
  \bibinfo{author}{\bibfnamefont{S.}~\bibnamefont{Coppersmith}},
  \bibnamefont{and} \bibinfo{author}{\bibfnamefont{L.~P.}
  \bibnamefont{Kadanoff}}, in \emph{\bibinfo{booktitle}{Perspectives and
  Problems in Nonlinear Science}}, edited by
  \bibinfo{editor}{\bibfnamefont{E.}~\bibnamefont{Kaplan}},
  \bibinfo{editor}{\bibfnamefont{J.~E.} \bibnamefont{Marsden}},
  \bibnamefont{and} \bibinfo{editor}{\bibfnamefont{K.~R.}
  \bibnamefont{Sreenivasan}} (\bibinfo{publisher}{Springer},
  \bibinfo{address}{New York}, \bibinfo{year}{2003}),
  chap.~\bibinfo{chapter}{2}, pp. \bibinfo{pages}{23--90},
  \bibinfo{note}{\url{http://arxiv.org/abs/nlin/0204062}}.

\bibitem[{\citenamefont{Stauffer and Aharony}(1994)}]{stauffer1994a}
\bibinfo{author}{\bibfnamefont{D.}~\bibnamefont{Stauffer}} \bibnamefont{and}
  \bibinfo{author}{\bibfnamefont{A.}~\bibnamefont{Aharony}},
  \emph{\bibinfo{title}{Introduction to Percolation Theory}}
  (\bibinfo{publisher}{Taylor and Francis}, \bibinfo{address}{London},
  \bibinfo{year}{1994}), \bibinfo{edition}{2nd} ed.

\bibitem[{\citenamefont{Newman}(2003)}]{newman2003a}
\bibinfo{author}{\bibfnamefont{M.~E.~J.} \bibnamefont{Newman}},
  \bibinfo{journal}{SIAM Review} \textbf{\bibinfo{volume}{45}},
  \bibinfo{pages}{167} (\bibinfo{year}{2003}).

\bibitem[{\citenamefont{Schneidman et~al.}(2006)\citenamefont{Schneidman,
  Berry, Segev, and Bialek}}]{schneidman2006a}
\bibinfo{author}{\bibfnamefont{E.}~\bibnamefont{Schneidman}},
  \bibinfo{author}{\bibfnamefont{M.~J.} \bibnamefont{Berry}},
  \bibinfo{author}{\bibfnamefont{R.}~\bibnamefont{Segev}}, \bibnamefont{and}
  \bibinfo{author}{\bibfnamefont{W.}~\bibnamefont{Bialek}},
  \bibinfo{journal}{Nature} \textbf{\bibinfo{volume}{440}},
  \bibinfo{pages}{1007} (\bibinfo{year}{2006}).

\bibitem[{\citenamefont{Campbell et~al.}(2011)\citenamefont{Campbell, Yang,
  Albert, and Shea}}]{campbell2011a}
\bibinfo{author}{\bibfnamefont{C.}~\bibnamefont{Campbell}},
  \bibinfo{author}{\bibfnamefont{S.}~\bibnamefont{Yang}},
  \bibinfo{author}{\bibfnamefont{R.}~\bibnamefont{Albert}}, \bibnamefont{and}
  \bibinfo{author}{\bibfnamefont{K.}~\bibnamefont{Shea}},
  \bibinfo{journal}{Proceedings of the National Academy of Sciences}
  \textbf{\bibinfo{volume}{108}}, \bibinfo{pages}{197} (\bibinfo{year}{2011}).

\bibitem[{\citenamefont{Granovetter}(1978)}]{granovetter1978a}
\bibinfo{author}{\bibfnamefont{M.}~\bibnamefont{Granovetter}},
  \bibinfo{journal}{American Journal of Sociology}
  \textbf{\bibinfo{volume}{83}}, \bibinfo{pages}{1420} (\bibinfo{year}{1978}).

\bibitem[{\citenamefont{Granovetter and Soong}(1986)}]{granovetter1986a}
\bibinfo{author}{\bibfnamefont{M.}~\bibnamefont{Granovetter}} \bibnamefont{and}
  \bibinfo{author}{\bibfnamefont{R.}~\bibnamefont{Soong}},
  \bibinfo{journal}{Journal of Economic Behavior and Organization}
  \textbf{\bibinfo{volume}{7}}, \bibinfo{pages}{83} (\bibinfo{year}{1986}).

\bibitem[{\citenamefont{Galam}(2008)}]{galam2008a}
\bibinfo{author}{\bibfnamefont{S.}~\bibnamefont{Galam}},
  \bibinfo{journal}{International Journal of Modern Physics C}
  \textbf{\bibinfo{volume}{19}}, \bibinfo{pages}{409} (\bibinfo{year}{2008}).

\bibitem[{\citenamefont{Schelling}(1971)}]{schelling1971a}
\bibinfo{author}{\bibfnamefont{T.~C.} \bibnamefont{Schelling}},
  \bibinfo{journal}{Journal of Mathematical Sociology}
  \textbf{\bibinfo{volume}{1}}, \bibinfo{pages}{143} (\bibinfo{year}{1971}).

\bibitem[{\citenamefont{Schelling}(1973)}]{schelling1973a}
\bibinfo{author}{\bibfnamefont{T.~C.} \bibnamefont{Schelling}},
  \bibinfo{journal}{Journal of Conflict Resolution}
  \textbf{\bibinfo{volume}{17}}, \bibinfo{pages}{381} (\bibinfo{year}{1973}).

\bibitem[{\citenamefont{Watts}(2002)}]{watts2002a}
\bibinfo{author}{\bibfnamefont{D.~J.} \bibnamefont{Watts}},
  \bibinfo{journal}{PNAS} \textbf{\bibinfo{volume}{99}}, \bibinfo{pages}{5766}
  (\bibinfo{year}{2002}).

\bibitem[{\citenamefont{Dodds and Watts}(2004)}]{dodds2004a}
\bibinfo{author}{\bibfnamefont{P.~S.} \bibnamefont{Dodds}} \bibnamefont{and}
  \bibinfo{author}{\bibfnamefont{D.~J.} \bibnamefont{Watts}},
  \bibinfo{journal}{Phys. Rev. Lett.} \textbf{\bibinfo{volume}{92}},
  \bibinfo{pages}{218701} (\bibinfo{year}{2004}).

\bibitem[{\citenamefont{Dodds et~al.}(2013)\citenamefont{Dodds, Harris, and
  Danforth}}]{dodds2013a}
\bibinfo{author}{\bibfnamefont{P.~S.} \bibnamefont{Dodds}},
  \bibinfo{author}{\bibfnamefont{K.~D.} \bibnamefont{Harris}},
  \bibnamefont{and} \bibinfo{author}{\bibfnamefont{C.~M.}
  \bibnamefont{Danforth}}, \bibinfo{journal}{Phys. Rev. Lett.}
  \textbf{\bibinfo{volume}{110}}, \bibinfo{pages}{158701}
  (\bibinfo{year}{2013}), \bibinfo{note}{\url{http://arxiv.org/abs/1208.0255}}.

\bibitem[{\citenamefont{Simmel}(1957)}]{simmel1957a}
\bibinfo{author}{\bibfnamefont{G.}~\bibnamefont{Simmel}},
  \bibinfo{journal}{American Journal of Sociology}
  \textbf{\bibinfo{volume}{62}}, \bibinfo{pages}{541} (\bibinfo{year}{1957}).

\bibitem[{\citenamefont{Bollob\'{a}s}(2001)}]{bollobas2001a}
\bibinfo{author}{\bibfnamefont{B.}~\bibnamefont{Bollob\'{a}s}},
  \emph{\bibinfo{title}{Random Graphs}} (\bibinfo{publisher}{Cambridge
  University Press}, \bibinfo{year}{2001}), \bibinfo{edition}{2nd} ed.

\bibitem[{Note1()}]{Note1}
Note1, \bibinfo{note}{when there are isolated nodes, their degrees are 0, and
  thus the denominator in the master equation is zero and $D$ is singular. If
  the initial network contains isolated nodes, we set all entries in the
  corresponding rows of $T$ to zero.}

\bibitem[{\citenamefont{Gleeson}(2008)}]{gleeson2008a}
\bibinfo{author}{\bibfnamefont{J.~P.} \bibnamefont{Gleeson}},
  \bibinfo{journal}{Phys. Rev. E} \textbf{\bibinfo{volume}{77}},
  \bibinfo{pages}{046117} (\bibinfo{year}{2008}).

\bibitem[{\citenamefont{West}(2001)}]{west2001a}
\bibinfo{author}{\bibfnamefont{D.~B.} \bibnamefont{West}},
  \emph{\bibinfo{title}{Introduction to Graph Theory}}
  (\bibinfo{publisher}{Prentice Hall}, \bibinfo{address}{Upper Saddle River,
  NJ}, \bibinfo{year}{2001}), \bibinfo{edition}{2nd} ed.

\bibitem[{\citenamefont{Oliveira}(2010)}]{oliveira2009a}
\bibinfo{author}{\bibfnamefont{R.~I.} \bibnamefont{Oliveira}}
  (\bibinfo{year}{2010}), \bibinfo{note}{\url{http://arxiv.org/abs/0911.0600}},
  \eprint{0911.0600}.

\bibitem[{\citenamefont{Derrida and Pomeau}(1986)}]{derrida1986a}
\bibinfo{author}{\bibfnamefont{B.}~\bibnamefont{Derrida}} \bibnamefont{and}
  \bibinfo{author}{\bibfnamefont{Y.}~\bibnamefont{Pomeau}},
  \bibinfo{journal}{Europhys. Lett.} \textbf{\bibinfo{volume}{1}},
  \bibinfo{pages}{45} (\bibinfo{year}{1986}).

\bibitem[{\citenamefont{Melnik et~al.}(2011)\citenamefont{Melnik, Hackett,
  Porter, Mucha, and Gleeson}}]{melnik2011a}
\bibinfo{author}{\bibfnamefont{S.}~\bibnamefont{Melnik}},
  \bibinfo{author}{\bibfnamefont{A.}~\bibnamefont{Hackett}},
  \bibinfo{author}{\bibfnamefont{M.~A.} \bibnamefont{Porter}},
  \bibinfo{author}{\bibfnamefont{P.~J.} \bibnamefont{Mucha}}, \bibnamefont{and}
  \bibinfo{author}{\bibfnamefont{J.~P.} \bibnamefont{Gleeson}},
  \bibinfo{journal}{Phys. Rev. E} \textbf{\bibinfo{volume}{83}},
  \bibinfo{pages}{036112} (\bibinfo{year}{2011}).

\bibitem[{\citenamefont{Phillips}(2003)}]{phillips2003a}
\bibinfo{author}{\bibfnamefont{G.~M.} \bibnamefont{Phillips}},
  \emph{\bibinfo{title}{Interpolation and Approximation by Polynomials}}
  (\bibinfo{publisher}{Springer}, \bibinfo{year}{2003}), \bibinfo{note}{see
  chapter 7}.

\bibitem[{\citenamefont{{Pe\~{n}a}}(1999)}]{pena1999a}
\bibinfo{editor}{\bibfnamefont{J.~M.} \bibnamefont{{Pe\~{n}a}}}, ed.,
  \emph{\bibinfo{title}{Shape Preserving Representations in Computer-Aided
  Geometric Design}} (\bibinfo{publisher}{Nova Science Publishers, Inc},
  \bibinfo{year}{1999}).

\bibitem[{\citenamefont{Payne et~al.}(2011)\citenamefont{Payne, Harris, and
  Dodds}}]{payne2011a}
\bibinfo{author}{\bibfnamefont{J.~L.} \bibnamefont{Payne}},
  \bibinfo{author}{\bibfnamefont{K.~D.} \bibnamefont{Harris}},
  \bibnamefont{and} \bibinfo{author}{\bibfnamefont{P.~S.} \bibnamefont{Dodds}},
  \bibinfo{journal}{Phys. Rev. E} \textbf{\bibinfo{volume}{84}},
  \bibinfo{pages}{016110} (\bibinfo{year}{2011}),
  \urlprefix\url{http://link.aps.org/doi/10.1103/PhysRevE.84.016110}.

\bibitem[{\citenamefont{Dodds et~al.}(2011)\citenamefont{Dodds, Harris, and
  Payne}}]{dodds2011a}
\bibinfo{author}{\bibfnamefont{P.~S.} \bibnamefont{Dodds}},
  \bibinfo{author}{\bibfnamefont{K.~D.} \bibnamefont{Harris}},
  \bibnamefont{and} \bibinfo{author}{\bibfnamefont{J.~L.} \bibnamefont{Payne}},
  \bibinfo{journal}{Phys. Rev. E} \textbf{\bibinfo{volume}{83}},
  \bibinfo{pages}{056122} (\bibinfo{year}{2011}),
  \urlprefix\url{http://link.aps.org/doi/10.1103/PhysRevE.83.056122}.

\bibitem[{\citenamefont{Gleeson and Cahalane}(2007)}]{gleeson2007a}
\bibinfo{author}{\bibfnamefont{J.~P.} \bibnamefont{Gleeson}} \bibnamefont{and}
  \bibinfo{author}{\bibfnamefont{D.~J.} \bibnamefont{Cahalane}},
  \bibinfo{journal}{Phys. Rev. E} \textbf{\bibinfo{volume}{75}},
  \bibinfo{pages}{056103} (\bibinfo{year}{2007}),
  \urlprefix\url{http://link.aps.org/doi/10.1103/PhysRevE.75.056103}.

\bibitem[{\citenamefont{Alligood et~al.}(1996)\citenamefont{Alligood, Sauer,
  and Yorke}}]{alligood1996a}
\bibinfo{author}{\bibfnamefont{K.~T.} \bibnamefont{Alligood}},
  \bibinfo{author}{\bibfnamefont{T.~D.} \bibnamefont{Sauer}}, \bibnamefont{and}
  \bibinfo{author}{\bibfnamefont{J.~A.} \bibnamefont{Yorke}},
  \emph{\bibinfo{title}{Chaos: An Introduction to Dynamical Systems}}
  (\bibinfo{publisher}{Springer}, \bibinfo{year}{1996}).

\bibitem[{\citenamefont{{NIST}}(2012)}]{dlmf2012}
\bibinfo{author}{\bibnamefont{{NIST}}}, \emph{\bibinfo{title}{{Digital Library
  of Mathematical Functions}}} (\bibinfo{year}{2012}),
  \bibinfo{note}{\url{http://dlmf.nist.gov}}.

\bibitem[{\citenamefont{Feigenbaum}(1978)}]{feigenbaum1978a}
\bibinfo{author}{\bibfnamefont{M.~J.} \bibnamefont{Feigenbaum}},
  \bibinfo{journal}{Journal of Statistical Physics}
  \textbf{\bibinfo{volume}{19}}, \bibinfo{pages}{25} (\bibinfo{year}{1978}).

\bibitem[{\citenamefont{Feigenbaum}(1979)}]{feigenbaum1979a}
\bibinfo{author}{\bibfnamefont{M.~J.} \bibnamefont{Feigenbaum}},
  \bibinfo{journal}{Journal of Statistical Physics}
  \textbf{\bibinfo{volume}{21}}, \bibinfo{pages}{669} (\bibinfo{year}{1979}).

\bibitem[{\citenamefont{Pachucki et~al.}(2011)\citenamefont{Pachucki, Jacques,
  and Christakis}}]{pachucki2011a}
\bibinfo{author}{\bibfnamefont{M.~A.} \bibnamefont{Pachucki}},
  \bibinfo{author}{\bibfnamefont{P.~F.} \bibnamefont{Jacques}},
  \bibnamefont{and} \bibinfo{author}{\bibfnamefont{N.~A.}
  \bibnamefont{Christakis}}, \bibinfo{journal}{American Journal of Public
  Health} \textbf{\bibinfo{volume}{101}}, \bibinfo{pages}{2170}
  (\bibinfo{year}{2011}).

\bibitem[{\citenamefont{Centola}(2011)}]{centola2011a}
\bibinfo{author}{\bibfnamefont{D.}~\bibnamefont{Centola}},
  \bibinfo{journal}{Science} \textbf{\bibinfo{volume}{334}},
  \bibinfo{pages}{1269} (\bibinfo{year}{2011}).

\bibitem[{\citenamefont{Romero et~al.}(2011)\citenamefont{Romero, Meeder, and
  Kleinberg}}]{romero2011a}
\bibinfo{author}{\bibfnamefont{D.~M.} \bibnamefont{Romero}},
  \bibinfo{author}{\bibfnamefont{B.}~\bibnamefont{Meeder}}, \bibnamefont{and}
  \bibinfo{author}{\bibfnamefont{J.}~\bibnamefont{Kleinberg}}, in
  \emph{\bibinfo{booktitle}{Proc. 20th ACM International World Wide Web
  Conference}} (\bibinfo{year}{2011}).

\bibitem[{\citenamefont{Winkler et~al.}(1972)\citenamefont{Winkler, Roodman,
  and Britney}}]{winkler1972a}
\bibinfo{author}{\bibfnamefont{R.~L.} \bibnamefont{Winkler}},
  \bibinfo{author}{\bibfnamefont{G.~M.} \bibnamefont{Roodman}},
  \bibnamefont{and} \bibinfo{author}{\bibfnamefont{R.~R.}
  \bibnamefont{Britney}}, \bibinfo{journal}{Management Science}
  \textbf{\bibinfo{volume}{19}}, \bibinfo{pages}{290} (\bibinfo{year}{1972}).

\end{thebibliography}

\end{document}